\documentclass[aps,prl,twocolumn,showpacs,superscriptaddress,groupedaddress]{revtex4}
\usepackage{amsmath}
\usepackage{amsfonts}
\usepackage{amssymb}
\usepackage{graphicx}
\usepackage{morefloats}
\usepackage{dcolumn}   
\usepackage{bbold}
\usepackage{bm}       
\usepackage{multirow}
\usepackage{relsize}   
\usepackage{amsthm}

\begin{document}

\newcommand{\ket}[1]{\left |#1 \right \rangle}
\newcommand{\bra}[1]{\left  \langle #1 \right |}
\newcommand{\braket}[2]{\left \langle #1 \middle | #2 \right \rangle}
\newcommand{\ketbra}[2]{\left | #1 \middle \rangle \middle \langle #2 \right |}

\newtheorem{conj}{Conjecture}
\newtheorem{thm}{Theorem}

\widetext

\title{Constraints on Maximal Entanglement Under Groups of Permutations}
\author{Alexander Meill, Jayden Butts, and Elijah Sanderson}
\date{\today}

\begin{abstract}
We provide a simplified characterization of entanglement in physical systems which are symmetric under the action of subgroups of the symmetric group acting on the party labels.  Sets of entanglements are inherently equal, lying in the same orbit under the group action, which we demonstrate for cyclic, dihedral, and polyhedral groups.  We then introduce new, generalized relationships for the maxima of those entanglement by exploiting the normalizer and normal subgroups of the physical symmetry group.
\end{abstract}

\maketitle

\section{Introduction}

Few features of quantum mechanics have proved more challenging, interesting, and widely applicable than entanglement.  It has been connected to a wide range of fields of physics, from quantum field theory \cite{Maldacena1999}\cite{Almheiri:2014lwa} to condensed matter physics \cite{Wiegner:12}\cite{2014JSMTE..11..013L}, to quantum information \cite{PMID:32541968}\cite{PhysRevA.98.042338}.  Its remarkable applications have prompted major mathematical studies devoted entirely to understanding its properties \cite{doi:10.1063/1.1497700}\cite{Gour_2008}\cite{PhysRevA.72.032313}.  Entanglement has been found to come in many forms \cite{PhysRevA.65.052112}\cite{PhysRevLett.108.230502}, to be measurable and quantified by different means \cite{PhysRevA.72.052331}, and subject to constraints on how it can be created and distributed in various systems \cite{PhysRevA.61.052306}\cite{PhysRevA.69.022309}.  Understanding the nature and limits of entanglement continues to be an active and flourishing area of research, both theoretical and experimental.

A natural question regarding entanglement, which has drawn much attention, is how much entanglement a given system can support.  More specifically, what are the maximally entangled states of a given Hilbert space, and how much entanglement do they exhibit?  This question has been asked, studied, and answered in many types of physical systems, ranging from spin lattices \cite{Hastings_2007} to black holes \cite{Faulkner:2013ica}.  A limiting frustration in those studies is that the complexity of analytically quantifying entanglement grows rapidly as the number of particles and dimensions of the spaces associated to those particles increase.

One approach to reducing the analytic hurdle of working with entanglement is to restrict the Hilbert space to some smaller, physically relevant subsystem.  A common choice of subsystem are those which exhibit some physical symmetry.  Symmetries such as permutation, translation, and rotation invariance are ubiquitous in entanglement research \cite{PhysRevLett.95.260604}\cite{Ungar2018}.  Not only do they shrink the size of the Hilbert space, those symmetries form inherent relationships between the entanglements \cite{PhysRevA.96.062310}\cite{PhysRevA.100.042318}, making the evaluation of maximal entanglement much more tractable.

The study of the relationship between symmetries and entanglement naturally invokes many mathematical approaches.  In this work, we examine the group theoretic properties of permutations of party labels in a multipartite system.  More specifically, for some subgroup $G \leq S_n$ of the symmetric group acting as permutations on the party labels, we show that the noramlizer and normal subgroups of $G$ offer a powerful tool in simplifying the landscape of maximal entanglement of systems invariant under $G$.

In this paper we by introducing the problem of counting distinct entanglements under some permutation group symmetry, $G$.  Then we introduce the relationships between the maxima of those entanglement, which is the main result of this paper.  We then give examples of those relationships in some common settings, and finally we conclude.

\section{Background and Definitions}

This work considers the entanglement of $n$ particles, to which the dimension of the Hilbert space associated to each particle is $d$.  We write the states of those particles as
\begin{align}
    \ket \psi &= \sum_{i_1=1}^d \ldots \sum_{i_n=1}^d a_{i_1 \ldots i_n} \ket{i_1 \ldots i_n} \\
    &= \sum_{\bold{i} \in [d]^{\otimes n}} a_{\bold{i}} \ket{\bold{i}},
\end{align}
introducing the notation that a bolded character, $\bold{i}=(i_1 \ldots i_n)$, represents an element of $[d]^{\otimes n}$, where $[d]=\{1, \ldots, n\}$ is the set of integers from 1 to $d$.

These particles are assumed to live in a system which is symmetric under some set of physical transformations, such as rotations and reflections.  Those transformations permute the party labels, and can therefore be represented as a group, $G$, which is a subgroup of the symmetric group, $S_n$, acting on the $n$ party labels.  The permuting action of an element, $g \in G$, on a party label, $i \in \mathbb [n]$, is most simply written as $g(i)$, but we also extend that notation to lists of $n$ integers and associated state vectors,
\begin{align}
    g(\bold i) &= g(i_1 \ldots i_n) = \left(i_{g^{-1}(1)} \ldots i_{g^{-1}(n)} \right) \\
    U_g \ket{\bold{i}} &= \ket{g(\bold{i})},
\end{align}
where $U_g$ is the unitary representation of $g$ on $\mathbb C_d^{\otimes n}$.  The physical symmetry implies that states in the system are invariant, up to an overall phase, under the action of $G$,
\begin{align}\label{Gsym}
    U_g \ket{\psi} = e^{i \theta_g} \ket{\psi} \quad \forall \quad g \in G.
\end{align}
This constrains the state coefficients according to
\begin{align}
    U_g \ket{\psi} &= e^{i \theta_g} \ket{\psi} \\
    \sum_{\bold{i} \in [d]^{\otimes n}} a_{\bold{i}} U_g \ket{\bold{i}} &= \sum_{\bold{i} \in [d]^{\otimes n}} a_{\bold{i}} e^{i \theta_g} \ket{\bold{i}} \\
    \sum_{\bold{i} \in [d]^{\otimes n}} a_{g^{-1}(\bold{i})} \ket{\bold{i}} &= \sum_{\bold{i} \in [d]^{\otimes n}} a_{\bold{i}} e^{i \theta_g} \ket{\bold{i}} \\ \label{asym}
    \implies \quad a_{\bold i} &= a_{g(\bold i)} e^{i \theta_g} \quad \forall \quad \bold{i}\in [d]^{\otimes n}, g \in G.
\end{align}
We will denote the Hilbert space of states with the above symmetry as $\mathcal H_G^{(n)}$.

Within these symmetric physical systems we will be examining the entanglement of the reduced states of sets of $m$ particles.  Let $E$ be an entanglement monotone defined on mixed states of $m$ particles, and let $X_m$ be the set of disordered $m$-tuples of distinct elements of $[n]$.  We then evaluate $E$ as
\begin{align}
    E_{x \in X_m} \left( \ket{\psi} \right) = E \left( \rho_x \right) = E \left( \text{Tr}_{\bar x} \left( \ketbra \psi \psi \right) \right).
\end{align}
Importantly, because $x$ is a disordered set of party labels, $E$ should not depend on the ordering of the particles upon which it evaluates.  There are many such entanglement measures of interest.  For pairwise entanglement, $m=2$, the concurrence \cite{PhysRevLett.80.2245} and negativity \cite{PhysRevLett.77.1413} are common representatives among the large set of measures.  In $m=3$, the 3-tangle \cite{PhysRevA.61.052306} is a notable candidate, though, like many pure state measures, it requires the convex roof extension to be defined on mixed states.  And for larger, arbitrary $m$, the geometric measure of entanglement \cite{https://doi.org/10.1111/j.1749-6632.1995.tb39008.x} and generalized versions of the tangle \cite{PhysRevA.63.044301} offer an analytically challenging but well defined approach.

For the most naive counting, calculating every possible entanglement for a state, $\ket \psi \in H_G^{(n)}$, and for a given $m$ would result in $\binom{n}{m}$ entanglements.  This count is greatly reduced, however, by the symmetry under $G$ thanks to,
\begin{align}\label{Eorbits}
    E_{g(x)}(\ket{\psi})=E_x(U_{g^{-1}} \ket{\psi}) = E_x(\ket \psi) \quad \forall \quad g \in G.
\end{align}
This implies that any choice of parties which lies entirely in the orbit, $O_G(x)$, of $x$ under $G$ will lead to the same entanglement.  Determining the number of distinct entanglements now equates to finding the set of distinct orbits, $X_m/G$, and finding the order of that set, $|X_m/G|$.  The Cauchy-Frobenius Theorem offers one approach for finding that number,
\begin{align}
    \left| \frac{X_m}G \right| = \frac{1}{|G|} \sum_{g \in G} \left| X_m^g \right|,
\end{align}
where $X_m^g$ is the set of elements in $X_m$ that are fixed by $g$.  However, applications of this problem have been studied across many fields of mathematics, algebraic, combinatoric, and otherwise, and a pair of solutions are highlighted in the examples section of this work.

The general goal of this work is to continue this process of reducing the set of distinct entanglements.  It is worth noting at this point, that while our focus is on the entanglement of reductions of the overall state, $E(\rho_x)$, the equating of sets of parties, $x\in X_m$, implies that any form of entanglement which is uniquely specified by a choice of $x$ would be subject to the same reduced count according to (\ref{Eorbits}) as well as the first of our two theorems below.  The simplest example of this is that the complement of $x\in X_m$, which we have labeled $\bar x \in X_{n-m}$, is uniquely specified by $x$, as such $|X_m/G|=|X_{n-m}/G|$, and there are then the same number of distinct entanglement among $m$ and $n-m$ particles.  Another such example lies in the examination of pure state entanglement between bipartitions of the particles, namely the entanglement, $E_{x|\bar x}$ between the set of particles, $x$ and $\bar x$, such as those used in area law bounds for entanglement in condensed matter systems \cite{Hastings_2007}\cite{2017NatPh..13..556H}.

\section{Main Result - Constraints on Maximal Entanglements}

With the landscape of potential entanglements simplified by (\ref{Eorbits}), we turn our attention to the main subject of this work; finding relationships for the maxima of those entanglements.  At face value, for a given $m$, the number of distinct maximal entanglements to determine appears to be $|X_m/G|$.  We offer the following pair of theorems to reduce that count.

\subsubsection{Equating Maxima via the Normalizer of $G$}

Let $N_G$ be the normalizer of $G$ in $S_n$.  While the normalizer of $G$ may not convey any physical symmetry of the system, it does help reduce the number of distinct entanglements according to our first theorem,
\begin{thm}  The maximal entanglements evaluated on elements of $X_m$ are equated under the orbit of $N_G$.
\begin{align}
\max_{\ket \psi \in \mathcal H_G^{(n)}} E_x(\ket{\psi}) = \max_{ \ket\psi \in \mathcal H_G^{(n)}} E_{\nu(x)}(\ket{\psi}) \quad \forall \quad \nu \in N_G.
\end{align}
\end{thm}
\begin{proof}
Begin by considering some permutation, $\nu \in N_G$ and state, $\ket{\psi} \in \mathcal H_G^{(n)}$.  The transformed state, $U_\nu \ket \psi$ is not necessarily equal to $\ket \psi$, because $\nu$ is not necessarily an element of $G$.  We can, however, show that $U_\nu \ket \psi \in \mathcal H_G^{(n)}$, meaning the transformed state is still invariant under $G$.  This is simple to confirm,
\begin{align}
    U_g U_\nu \ket \psi = U_\nu U_{g'} \ket \psi = U_\nu \ket \psi,
\end{align}
where $g' \in G$ thanks to $\nu$ being in the normalizer of $G$, enabling the second equality via (\ref{Gsym}).  This then leaves us with
\begin{align}
    \max_{\ket \psi \in \mathcal H_G^{(n)}} E_x(\ket{\psi}) &= \max_{U_{\nu}\ket \psi \in \mathcal H_G^{(n)}} E_x(U_{\nu^{-1}}\ket{\psi}) \\ &= \max_{\ket \psi \in \mathcal H_G^{(n)}} E_{\nu(x)}(\ket{\psi}).
\end{align}
\end{proof}
This theorem offers another substantial reduction to the number of distinct maximal entanglements and the associated maximally entangled states.  Upon finding a state, $\ket{\phi}$, which maximizes $E_x$, it then follows that $U_\nu \ket \phi$ maximizes $E_{\nu(x)}$.  The obvious question, then is how many distinct maxima are left?  Answering that question may again pool from many areas of mathematics, depending on $G$, but we offer that the Cauchy-Frobenius Theorem again applies,
\begin{align}
    \left| \frac{X_m}{N_G} \right| = \frac{1}{|N_G|} \sum_{\nu \in N_G} \left| X_m^\nu \right|,
\end{align}
or equivalently, if $N_G/G$ and $X_m/G$ are the preferred group and set to work with,
\begin{align}
    \left| \frac{X_m}{N_G} \right| = \left| \frac{X_m/G}{N_G/G} \right|= \frac{1}{|N_G/G|} \sum_{\nu \in N_G/G} \left| X_m^\nu \right|.
\end{align}
Examples using the Cauchy-Frobenius Theorem as well as other, more direct approaches can be found in the examples section.

\subsubsection{Reducing to Normal Subgroups of $G$}

While the count of distinct maximal entanglements for a given $n$ and $m$ is at its apparent minimum, the following theorem offers a means of simplifying the entanglement on $x$ if each of its constituent parties lie in the same orbit of a normal subgroup of $G$.  Namely, for a normal subgroup, $H \lhd G$, we can observe that the orbits of party labels, $i \in [n]$, under $H$ form partitions of $[n]$ of equal order.  Let $Y \in [n]/H$ be one such orbit, and let $Y_m$ be the set of disordered $m$-tuples of distinct elements of $Y$.  With these definitions, we can introduce our second theorem,
\begin{thm}
For any normal subgroup, $H \lhd G$, any orbit $Y \in [n]/H$ under that subgroup, and choice of party labels, $x \in Y_m$, from the same orbit, we have that,
\begin{align}
    \max_{\ket \psi \in \mathcal H_G^{(n)}} E_x \left( \ket \psi \right) = \max_{\ket \phi \in \mathcal H_{H|Y}^{(|Y|)}} E_x \left( \ket \phi \right),
\end{align}
where $H|Y$ is the group of actions of $H$ restricted to $Y$.
\end{thm}
\begin{proof}
Consider $\rho_Y$, the reduced state of $\ket \psi \in \mathcal H_G^{(n)}$ for the particles in the orbit, $Y \in [n]/H$,
\begin{align}
    \rho_Y^{\vphantom{\dagger}} &= \text{Tr}_{\bar Y} \left( \ketbra \psi \psi \right) \\ &= \sum_{\bold{i},\bold{j} \in [d]^{\otimes n}} \; \sum_{\bold{k} \in [d]^{\otimes n-|Y|}} \bra{\bold k}a_{\bold{i}}^{\vphantom{*}} \ketbra{\bold i}{\bold j} a_{\bold{j},}^* \ket{\bold k},
\end{align}
where it is implied that the components of $\bold k$ are associated to the parties in $\bar Y$, which are being traced over.  For any $h \in H|Y$, and its complement, $\bar h \in H|\bar Y$, such that $h \times \bar h \in H$, we can show that,
\begin{align}
    U_h^{\vphantom{\dagger}} \rho_Y^{\vphantom{\dagger}} U_h^\dagger &= \sum_{\bold i, \bold j} \sum_{\bold k} a_{\bold{i}}^{\vphantom{*}} \bra{\bold k} U_h^{\vphantom{\dagger}} \otimes \mathbb 1_{\bar Y}^{\vphantom{\dagger}} \ketbra{\bold i}{\bold j} U_h^\dagger \otimes \mathbb 1_{\bar Y}^{\vphantom{\dagger}} \ket{\bold k} a_{\bold{j}}^* \\
    &= \sum_{\bold i, \bold j} \sum_{\bold k} a_{\bold{i}}^{\vphantom{*}} \bra{\bold k} \left(\mathbb 1_{\bar Y}^{\vphantom{\dagger}} \otimes U_{\bar h}^\dagger \right) U_{h\times \bar h}^{\vphantom{\dagger}} \ket{\bold i} \\
    \notag &\quad \quad \quad \quad \quad \times \bra{\bold j} U_{h\times \bar h}^{\dagger} \left(\mathbb 1_{\bar Y}^{\vphantom{\dagger}} \otimes U_{\bar h}^{\vphantom{\dagger}} \right) \ket{\bold k} a_{\bold j}^* \\
    &= \sum_{\bold i, \bold j} \sum_{\bold k} a_{h^{-1} \times {\bar h}^{-1}(\bold i)}^{\vphantom{*}} \braket{{\bar h}^{-1}(\bold k)}{\bold i} \\ &\quad \quad \quad \quad \quad \times \braket{\bold j}{{\bar h}^{-1}(\bold k)} a_{h^{-1} \times {\bar h}^{-1}(\bold j)}^* \\
    &= \rho_Y^{\vphantom{\dagger}},
\end{align}
where the last equality holds thanks to $(\ref{asym})$ and the freedom to reorder the sum in $\bold k$.  Now because $\rho_Y$ commutes with $U_h$, we can express it as
\begin{align}
    \rho_Y = \sum_l p_l \ketbra{\phi_l}{\phi_l},
\end{align}
where $\ket{\phi_l}$ are eigenstates of $U_h$.  Of course, $U_h$ being unitary implies that its eigenvalues are phases, so $U_h \ket{\phi_l} = e^{i \phi_{l,h}} \ket{\phi_l}$.  And this has so far been true for arbitrary $h \in H|Y$, meaning in total that each of the $\ket{\phi_l} \in \mathcal H_{H|Y}^{(|Y|)}$.

Now introduce $E_x$ for $x \in Y_m$.  Let the sum over $l$ in $\rho_Y$ be reordered in non-increasing order of $E_x(\phi_l)$, so $E_x(\phi_l) \geq E_x(\phi_{l+1})$.  We can now exploit the convexity of $E$ to show that,
\begin{align}
    E_x \left( \rho_Y \right) &\leq \sum_l p_l E_x \left( \ketbra{\phi_l}{\phi_l} \right) \\
    &\leq E_x \left( \ket{\phi_1} \right) \\
    &\leq \max_{\ket \phi \in \mathcal H_{H|Y}^{(|Y|)}} E_x \left( \ket \phi \right).
\end{align}

The above has been shown to be true for arbitrary $\ket \psi \in \mathcal H_G^{(n)}$, but it remains to be shown that the inequality can be saturated.  In doing so, we will need to construct the state, $\ket \Psi$, which achieves that maximum.  Start by identifying a state, $\ket \Phi$, which achieves the maximum entanglement in the reduced space,
\begin{align}
    E_x(\ket \Phi) = \max_{\ket \phi \in \mathcal H_{H|Y}^{(|Y|)}} E_x \left( \ket \phi \right).
\end{align}
We can now show that a viable construction for a state which provides the desired saturation is
\begin{align}
    \ket \Psi = \bigotimes_{Y \in [n]/H} \ket{\Phi}_Y.
\end{align}
It is apparent that $E_x(\ket \Psi) = E_x(\ket \Phi)$, but we have to show that $\ket \Psi \in \mathcal H_G^{(n)}$.  Thankfully, the action of $G$ on $[n]/H$ merely permutes entire orbits, for if parties $i\neq j$ are in the same orbit, then the same is true of $g(i)$ and $g(j)$.  So for arbitrary, $g \in G$,
\begin{align}
    U_g \ket \Psi &= \bigotimes_{Y \in [n]/H} \ket{\Phi}_{g(Y)},
\end{align}
which is again equal to $\ket{\Psi}$ thanks to the freedom to reorder the repeated tensor product in $Y$.  This resolves that $\ket{\Psi} \in \mathcal H_G^{(n)}$ and provides the desired maximum.
\end{proof}

While restricted to $x \in Y_m$, this theorem provides a powerful simplification to determining the maximal entanglement for certain $x$ by vastly reducing the size of the Hilbert space to maximize over.  It also gives a novel prescription for building entangled states, $\ket \Psi$, in the overall space by weaving together entangled states, $\ket \Phi$, from the reduced space.  It also follows that $\ket \Psi$ maximizes the entanglement not only for $x$, but for all $g(x)$.

It is also notable that this theorem simultaneously applies to entanglement among multiple numbers of particles, $m$.  So long as $x \in Y_m$, the theorem is viable, meaning it applies for $2 \leq m \leq |Y|$ for any given $H$.

\section{Examples}

We now present examples of how the entanglement landscape in $\mathcal H_G^{(n)}$ can be reduced by the results of this paper for various $G$ of notable physical significance.

\subsubsection{The Dihedral Group, $D_n$}

The Dihedral group, $D_n$, of symmetries of a regular $n$-gon is a natural starting point for the study of physical symmetries.  It is straightforward to construct a physical system of $n$ particles on a ring and ask that the reflections and rotations described by $D_n$ leave the system invariant.  The entanglement properties of such systems has been investigated previously \cite{Mirzaee2007}\cite{9063519}.  As we will see, though, with regards to maximal entanglement, the rotations described by the cyclic group, $C_n$, offer a complete picture on their own.

For an $n$ particle system invariant under $D_n$, we begin by considering $|X_m/D_n|$, the number of distinct entanglements of $m$ particles.  This problem has been solved previously \cite{Gupta79} in the context of combinatorics, where $|X_m/D_n|$ was contextualized as the number of distinct polygons formed by connecting $m$ out of $n$ points which evenly divide the circumference of a circle.  The solution was found to be
\begin{align}\label{Gupta}
    \left| \frac{X_m}{D_n} \right| = \frac 12 \left( \binom{\lfloor \frac{n-h_m}2 \rfloor}{\lfloor \frac m2 \rfloor}  + \frac 1m \sum_{\delta | (m,n)} \varphi(\delta) \binom{\frac n \delta -1}{\frac m \delta -1}\right),
\end{align}
where $h_m=m \mod 2$ and $\varphi(n)$ is Euler's totient function.

While the complete setting of $\mathcal H_{D_n}^{(n)}$ is a rich subject for the study of entanglement in physical systems, the study of maximal entanglements proves to take a narrower scope.  Thanks to the fact that $C_n \lhd D_n$, and that $C_n$ also acts transitively on $[n]$, Theorem 2 allows us to conclude that $\mathcal H_{C_n}^{(n)}$ and $\mathcal H_{D_n}^{(n)}$ share the same maxima.  As such, we will proceed to instead consider $C_n$.

\subsubsection{The Cyclic Group, $C_n$}

Let us instead consider the cylcic group, $C_n$, of order, $n$.  The symmetry of $C_n$ manifests physically in many ways, most commonly as translational invariance for chains of particles with periodic boundary conditions, or, equivalently, particles arranged on a ring subject to rotation.  Such configurations have received ample attention in the study of condense matter systems \cite{PhysRevB.100.035113} and in quantum information \cite{PhysRevLett.123.110502}.  The entanglement of these systems has accordingly been studied on many occasions \cite{PhysRevA.100.042318}\cite{PhysRevA.63.052302}.  As such, this simple setting is a prime test case for our work.

We again begin by considering the number of distinct entanglements, $|X_m/C_n|$.  This problem has also been solved previously in the context of combinatorics, as $|X_m/C_n|$ also corresponds to the number of distinct 2-color necklaces of $n$ beads, with $m$ of one color and $n-m$ of the other.  The solution provided by \cite{shevelev2011problem} is
\begin{align}\label{Shev}
    \left | \frac{X_m}{C_n} \right| = - \sum_{\delta \leq 2, \delta |(n,m)} \mu(\delta) \left | \frac{X_{m/\delta}}{D_{n/\delta}} \right|,
\end{align}
where $\mu(n)$ is the Mobius function and $|X_m/D_n|$ is exactly that from (\ref{Gupta}).

To illustrate the formula resulting from $(\ref{Shev})$, let us demonstrate the sets of entanglements in $n=8$ particles which are equated by $C_8$.  These are shown for $m=2$, 3, and 4 below in Figure 1.  Recall that the picture is much the same for sets of $n-m$ particles.
\begin{widetext}
\begin{center}
\begin{figure}[h]
\centering
\includegraphics[width = 170mm]{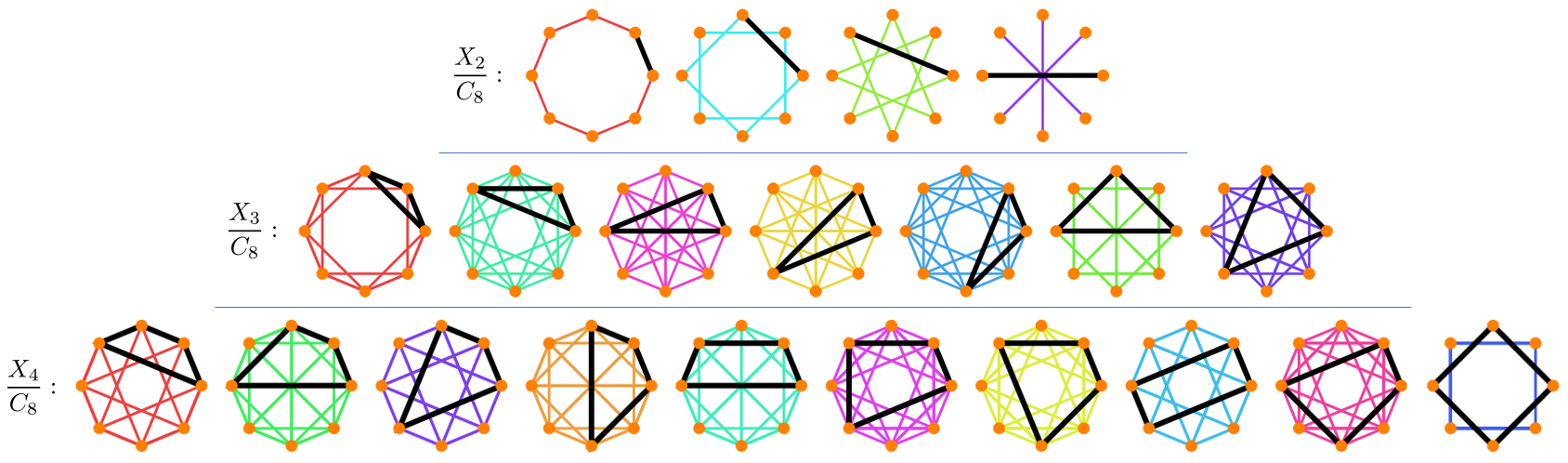}
\caption{The sets of distinct entanglements in $\mathcal H_{C_8}^{(8)}$.  Each element of $X_m/C_8$ shows the equated elements of $X_m$ with a single element highlighted.}
\end{figure}
\end{center}
\end{widetext}

While $(\ref{Shev})$ still leaves us with a formidable number of entanglements, that picture simplifies when we consider the maxima of those entanglements.  We can begin by invoking Theorem 1, which states that any two entanglements which lie in the same orbit of the normalizer of $C_n$ will share the same maximum.  The normalizer, $N_{C_n}$, of $C_n$ takes a simple form when expressed by its action on the party labels in $[n]$.  In particular, $N_{C_n}$ acting on $[n]$ is the set of permutations of the form,
\begin{align}
   \nu(i) = \alpha \, i + \beta \mod n,
\end{align}
where $\alpha, \beta \in [n]$ and $gcd(\alpha,n)=1$.  We then have that $|N_{C_n}|=n \, \varphi(n)$.  Equivalently, we could express $N_{C_n}/C_n$ simply by the $\alpha$ which `spread' the party indices.  In either case, now have a means of equating maximal entanglements and can turn to finding the number of remaining distinct maxima, $|X_m/N_{C_n}|$.

An expression for $|X_m/N_{C_n}|$ is not known for general $n$ and $m$, and is left as an open problem.  The case of $|X_2/N_{C_n}|=\tau(n)-1$ was shown in \cite{PhysRevA.100.042318}, where $\tau(n)$ is the number of divisors of $n$.  This can be seen by showing that any element, $x\in X_2$, can be mapped by $N_{S_n}$ to $x=(i,i+\delta)$ for some $\delta|n$, with the exception that $\delta \neq n$.

Despite not having a general solution, we can still understand $X_m/N_{C_n}$ in small examples, such as in $\mathcal H_{C_8}^{(8)}$, as demonstrated in Figure 2 below.

\begin{widetext}
\begin{center}
\begin{figure}[h]
\centering
\includegraphics[width = 110mm]{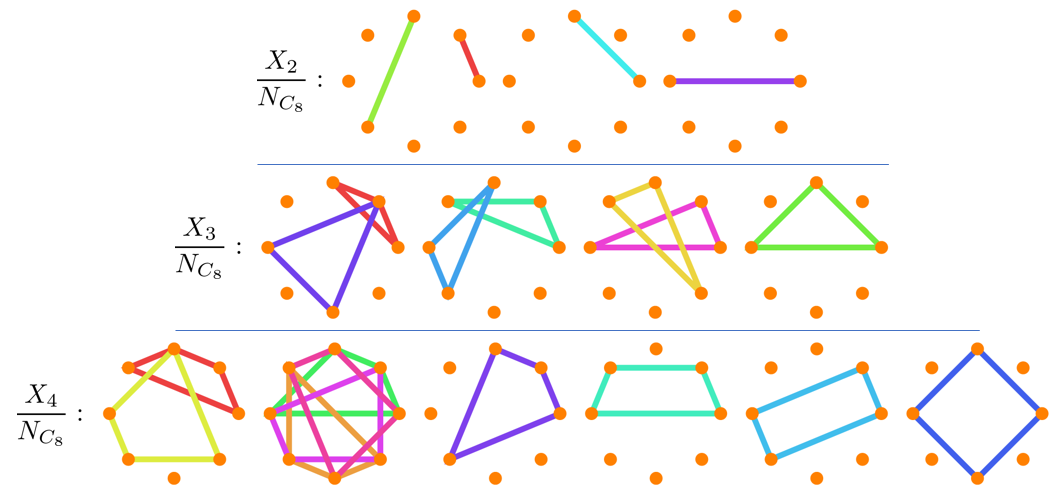}
\caption{The sets of distinct maximal entanglements in $\mathcal H_{C_8}^{(8)}$ identified from distinct orbits under $N_{C_8}$.  Each element of $X_m/N_{C_8}$ shows the elements of $X_m/C_8$ whose maxima were equated under $N_{C_8}$.}
\end{figure}
\end{center}
\end{widetext}

As we can see, the action of the normalizer took us from 4, 7, and 10 distinct maxima for $m=2$, 3, and 4 respectively, down to 3, 4, and 6.  There is still more to simplify, though, because Theorem 2 shows that some of these maxima are not unique to $C_8$, and can be restricted to a simpler problem.

In using Theorem 2, we first identify normal subgroups of $C_n$.  Of course, since $C_n$ is abelian, all of its subgroups are normal are isomorphic to $C_k$ where $k|n$, which confirms the generalized result in \cite{PhysRevA.100.042318} that
\begin{align}
    \max_{\ket \psi \in \mathcal H_{C_n}^{(n)}} E_x \left( \ket \psi \right) = \max_{\ket \phi \in \mathcal H_{C_k}^{(n/k)}} E_x \left( \ket \phi \right)
\end{align}
where the components, $x_i$, of $x$ differ by multiples of $k$.  This highlights that the cyclic group has the convenient property that the maximally entangled states in large systems can be built from maximally entangled states in smaller systems with the same symmetry.  In any given, $n$, then, there are some smaller set of entanglements which are unique to that $n$.  For example, in $n=8$, Figure 3 shows the entanglements which can be extended from $\mathcal H_{C_2}^{(2)}$ and $\mathcal H_{C_4}^{(4)}$.  This results in only 1, 3, 5 maximal entanglements which are unique to $\mathcal H_{C_8}^{(8)}$ for $m=2$, 3, and 4 respectively.

\begin{widetext}
\begin{center}
\begin{figure}[h]
\centering
\includegraphics[width = 110mm]{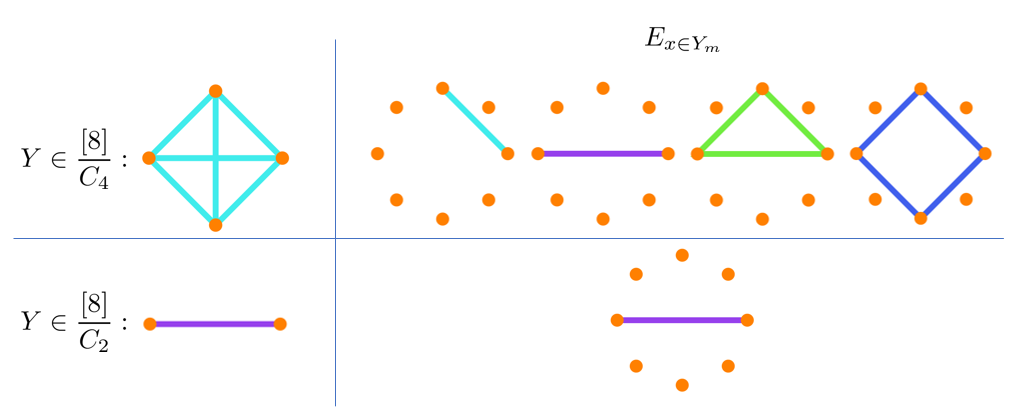}
\caption{Sample orbits of $[8]$ under the normal subgroups of $C_8$, as well as the entanglements which fall under that orbit and whose maxima will be subject to the simplification of Theorem 2.}
\end{figure}
\end{center}
\end{widetext}

\subsubsection{The Polyhedral Groups, $T$, $O$, and $I$}

The natural extension of our previous examples is to consider symmetries which arise from three dimensional arrangements of particles.  In keeping with the highly symmetric picture, we will consider particles configured at the vertices of the platonic solids.  The set of rotations and reflections which leave the platonic solids invariant from the polyhedral groups, $T$, $O$, and $I$.  Similar to the rings of particles under $C_n$ and $D_n$, the action of the polyhedral groups on their corresponding arrangements of particles permute those party labels transitively.  These systems offer another interesting application for the study of entanglement due to their relevance in molecular chemistry and again in condensed matter physics \cite{PhysRevB.88.054101}.

The simplest place to begin is with the tetrahedron and the tetrahedral group, $T$.  This case is somewhat trivial, though, as $T$ acting on $[n=4]$ leaves only a single orbit in $X_2$, and therefore a single entanglement with no relations to make.

Moving up in size, we can now consider the octahedral group, $O$, which acts on both the octahedron and the cube.  We can start with the octahedron and find the orbits of $X_m/O_6$ for $m=2$ and $m=3$, where we have added the subscript, $n$, to $O$ to distinguish between the octahedron, $n=6$, and cube, $n=8$.  These sets of distinct entanglements are shown in Figure 4.

\begin{center}
\begin{figure}[h]
\centering
\includegraphics[width = 50mm]{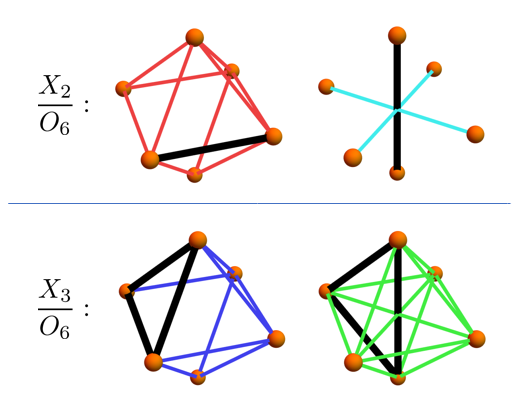}
\caption{The sets of distinct entanglements in $\mathcal H_{O_6}^{(6)}$ with highlighted representative from each set.}
\end{figure}
\end{center}

Turning now to the maxima of these entanglements, we can again make use of Theorem 2 by identifying that the inversion operation in $O_6$, which exchanges each particle with its spatially furthest counterpart, forms, together with the identity, a normal subgroup isomorphic to $C_2$.  The orbits of this subgroup are of course the pairs of opposite particles, as demonstrated in Figure 5.

\begin{center}
\begin{figure}[h]
\centering
\includegraphics[width = 50mm]{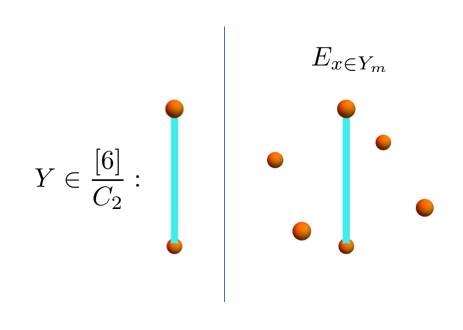}
\caption{A sample orbit of $[n=6]$ under $C_2 \lhd O_6$ and the corresponding entanglement element of that orbit.}
\end{figure}
\end{center}

It is convenient that these orbits involve only two particles, for which pure state entanglement is relatively well understood.  For example, one could construct the state which maximizes opposing pairwise entanglement in $\mathcal H_{O_6}^{(6)}$ by placing Bell-type pairs across from each other on each diameter of the octahedron.

We can offer a similar treatment of the octahedral group symmetry on the cube, $O_8$, starting by identifying the distinct entanglements via the orbits in $X_m/O_8$ for $m=2$, 3, and 4.  These are compiled in Figure 6.

\begin{widetext}
\begin{center}
\begin{figure}[h]
\centering
\includegraphics[width = 110mm]{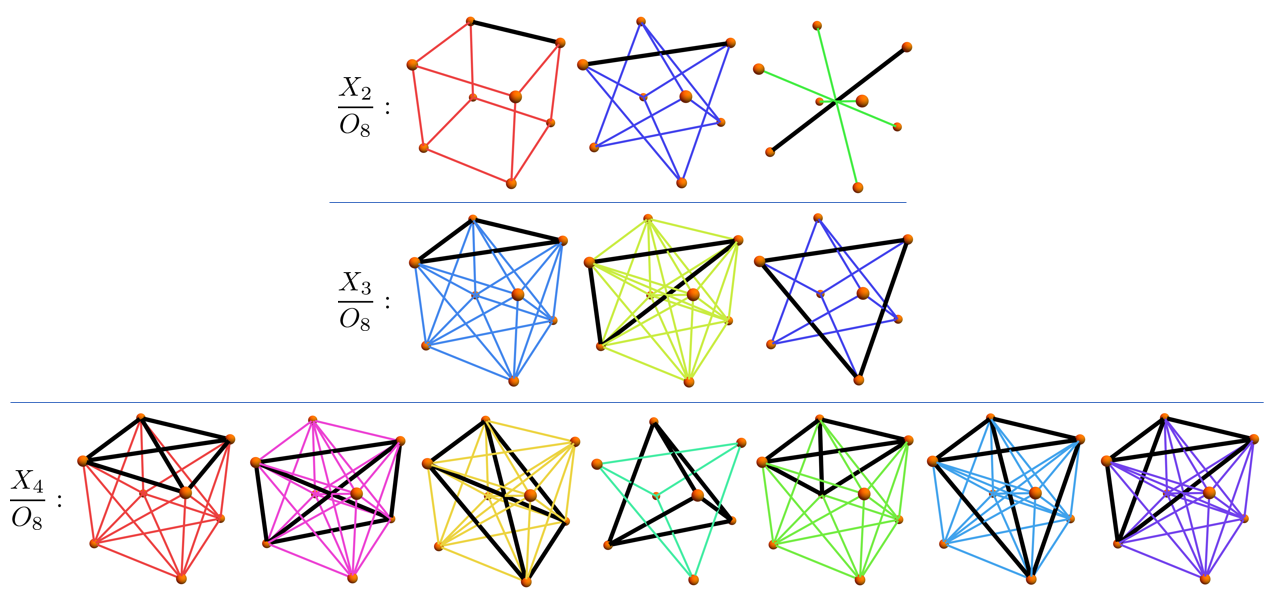}
\caption{The sets of distinct entanglements in $\mathcal H_{O_8}^{(8)}$ with highlighted representative from each set.}
\end{figure}
\end{center}
\end{widetext}

In simplifying the maxima of these entanglements, we can again identify the group generated by the inversion operation as a normal subgroup isomorphic to $C_2$.  Additionally, now, we have a second normal subgroup provided by the symmetries of the tetrahedral group, $T$.  This removes a number of entanglements to maximize over in $\mathcal H_{O_8}^{(8)}$, as diagrammed in Figure 7.

\begin{center}
\begin{figure}[h]
\centering
\includegraphics[width = 80mm]{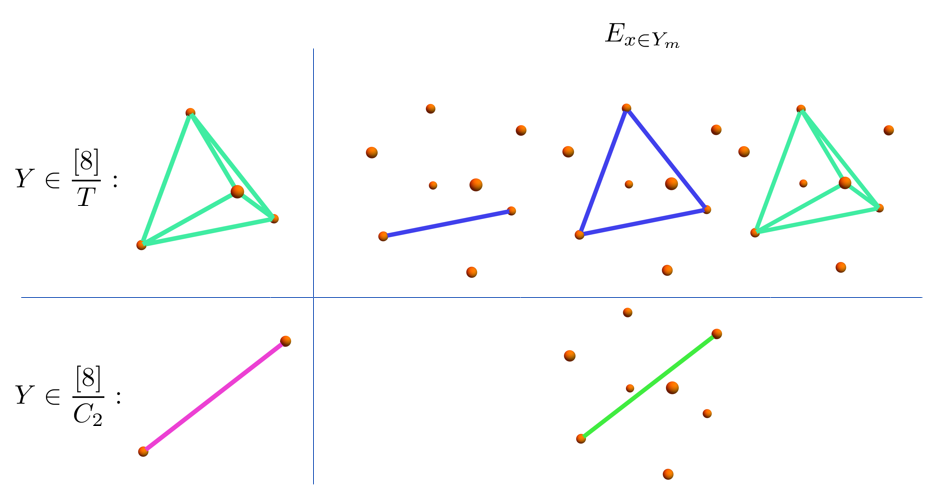}
\caption{Sample orbits of $[n=8]$ under $C_2, T \lhd O_8$ and the corresponding entanglement elements within those orbits.}
\end{figure}
\end{center}

As an interesting demonstration of this restriction, consider evaluating the $m=2$ pairwise concurrence \cite{PhysRevLett.80.2245} of particles on opposite corners of the face of the cube.  As we have shown, the maximum of this entanglement is the same as the maximum for pairs on the tetrahedron.  This maximum is known, thanks to the fact that the action of $T$ on $n=4$ is isomorphic to that of $S_4$.  And for such totally permutation invariant states, the maximal pairwise concurrence is known \cite{PhysRevA.62.050302} to be $2/n$, achieved by the generalized $W$-state,
\begin{align}
    \ket{W_n}=\frac{1}{\sqrt n} \sum_{\pi \in C_n} U_\pi | 1 \underbrace{0 \ldots 0}_{n-1} \rangle.
\end{align}
All together, then, we can construct the state,
\begin{align}
    \ket{\Psi}=\ket{W_4}_{Y_1} \otimes \ket{W_4}_{Y_2},
\end{align}
where $Y_j$ represent sets of 4 particles on opposite face corners, and conclude that this state achieves the associated maximal pairwise concurrence of $1/2$.

It is interesting to note that Theorem 1 was of no aid in the study of the octahedron and cube because of the trivial nature of the action of the normalizers of $O_6$ and $O_8$ on $[n=6]$ and $[n=8]$.  The larger solids, however, do indeed have non-trivial normalizers which allows us to invoke Theorem 1.  Let us demonstrate this effect on the icosahedron, which is symmetrized by the icosahedral symmetry group on 12 vertices, $I_{12}$.  Let us begin by examining pairwise entanglements in $X_2/I_{12}$, for which there are only the three elements shown in Figure 8.

\begin{center}
\begin{figure}[h]
\centering
\includegraphics[width = 80mm]{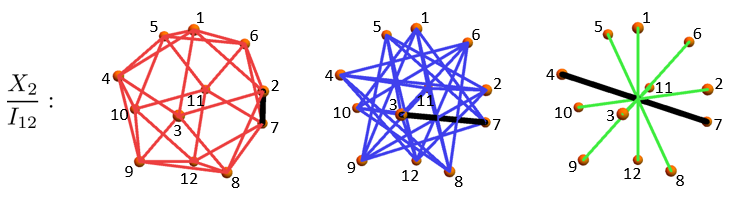}
\caption{The set of orbits, $X_2/I_{12}$, with representative element highlighted and vertices labeled.}
\end{figure}
\end{center}

Examining the order of the normalizer of $I_{12}$, we find that it is twice the order of $I_{12}$ itself.  Using the vertex numbering shown in Figure 8, we can identify the permutation, $\nu = (2 \, 8 \, 6 \, 11)(3 \, 10 \, 5 \, 9)(4 \, 7)$, as a sample application of the non trivial action of that normalizer.  Applying $\nu$ exchanges the first two elements of $X_2/I_{12}$, therefore we can equate the maxima of those two entanglements by Theorem 1.  Of course, maximal entanglements for elements of $X_m/I_{12}$ for $m>2$ are likewise related by $\nu$ and Theorem 1, but those orbits are too unweildy to compile in full here, so we will suffice with demonstrating this result on $m=2$.  As a final note, we can observe the application of Theorem 2 again on $m=2$ thanks to $I_{12}$ having a normal subgroup isomorphic to $C_2$. This naturally reduces the third element of $X_2/I_{12}$ to isolated entangled pairs, allowing us to again achieve maximal entanglement by weaving bell pairs across antipodal vertices of the icosahedron.

\section{Conclusion}

We have offered two novel approaches to reducing the number of distinct maximal entanglements in a quantum system which is invariant under a group of permutations, $G$.  This significantly reduces the search for the set of unique maxima, but sadly offers little help in performing those remaining maximizations.  Maximizations for certain choices of $G$, $E$, $n$, and $m$ have been accomplished \cite{PhysRevA.96.062310}\cite{PhysRevA.62.050302}\cite{2010NJPh...12g3025A} and as that list grows, so too does the application of our theorems.

The choice of $G$ and $E$ in general is highly motivated by the nature of the physical system being considered.  We chose to restrict to $E$ which treat each party symmetrically and $G$ which act transitively on $[n]$, but it would be an interesting future direction to relax those constraints.  Doing so would add context to systems where the particles have a distinguishable structure, for instance a ring of particles with alternating Hilbert space dimension, $d$, or molecules made up of multiple species of atoms.  Here one might ask that the structure of $G$ reflect the construction of the physical systems by restricting the orbits of $G$ to particles of the same type.

An altogether new direction would be to consider symmetries other than those of permuting the party labels.  One might ask how the ideas of this paper could describe systems subject to some other group of symmetries, such as ones which act non-trivially on a single tensor factor, such as those considered in \cite{meill2019mean}.

\bibliography{mybib}

\end{document}